\documentclass[letterpaper, 10 pt, conference]{ieeeconf}
\usepackage{amssymb}
\usepackage[linesnumbered,ruled,vlined]{algorithm2e}
\usepackage{algorithmic}
\usepackage{multirow}
\usepackage{wrapfig}
\usepackage{hyperref}
\usepackage{xcolor}
\usepackage{times}

\usepackage{amsmath}
\usepackage{graphicx} 
\newenvironment{proofp}{\emph{Proof of Proposition 1:}}{\hfill$\square$}
\newenvironment{proofpp}{\emph{Proof of Proposition 2:}}{\hfill$\square$}
\newcommand{\para}[1]{\medskip \noindent \textbf{{#1}.}}

\hypersetup{
  colorlinks,
  citecolor=black,
  linkcolor=black,
  urlcolor=black}

\definecolor{orange(sae/ece)}{rgb}{1.0, 0.49, 0.0}
\definecolor{teal(sae/ece)}{rgb}{0, 0.47, 0.52}
\definecolor{purple}{rgb}{0.74, 0.65, 1.0}
\definecolor{light_gray}{rgb}{0.9, 0.9, 0.9}
\definecolor{medium_gray}{rgb}{0.6, 0.6, 0.6} 
\definecolor{dark_gray}{rgb}{0.2, 0.2, 0.2} 
\definecolor{dark_blue}{rgb}{0.098, 0.239, 0.52}
\definecolor{dark_brown}{rgb}{0.3255, 0.004, 0.001}
\definecolor{r3mcolor}{rgb}{0.478, 0.1569, 0.4863}
\definecolor{light_blue}{rgb}{0.33, 0.80, 1}

\newcommand{\ours}{{{Algorithm 1}}\xspace}
\newcommand{\passive}{{{passive}}\xspace}
\newcommand{\oracle}{{{complete information}}\xspace}
\newcommand{\revise}[1]{\color{black}{#1}}

\newcommand{\finalrevise}[1]{\color{black}{#1}}

\newcommand{\T}{\mathbf{T}}
\newcommand{\N}{\mathbf{N}}
\newcommand{\est}{\hat{\theta}}

\newcommand{\param}{\theta}

\newcommand{\stepsize}{\alpha}
\newcommand{\DemoCost}{c^{\mathrm{demo}}}
\newcommand{\TeachingCost}{\bar{c}}
\newcommand{\TeachingPolicy}{\bar{\pi}_t^1}
\newcommand{\TeachingPolicyUnLabeled}{\bar{\pi}}
\newcommand{\policy}{\pi}
\newcommand{\trueparam}{\theta^*}
\newcommand{\intentweight}{\rho_2}
\newcommand{\taskweight}{\rho_1}
  
  \newtheorem{theorem}{Theorem}

  \newtheorem{proposition}[theorem]{Proposition}
  \newtheorem{remark}[theorem]{Remark}

\IEEEoverridecommandlockouts
\begin{document}

\title{
Intent Demonstration in General-Sum Dynamic Games
\\via Iterative Linear-Quadratic Approximations}

\author{
Jingqi Li, Anand Siththaranjan, Somayeh Sojoudi, Claire Tomlin, Andrea Bajcsy
\thanks{
Jingqi Li is with the Oden Institute, University of Texas at Austin, TX 78712, USA (\href{mailto:jingqi.li@austin.utexas.edu}{jingqi.li@austin.utexas.edu}). Anand Siththaranjan is with the Department of Economics, Stanford University, CA 94305, USA (email: \href{mailto:anandks@stanford.edu}{anandks@stanford.edu}). 
Somayeh Sojoudi and Claire Tomlin are with the Department of Electrical Engineering and Computer Sciences, University of California, Berkeley, CA 94720, USA (email: \href{mailto:sojoudi@berkeley.edu}{sojoudi@berkeley.edu}, \href{mailto:tomlin@berkeley.edu}{tomlin@berkeley.edu}). Andrea Bajcsy is with the School of Computer Science, Carnegie Mellon University, PA, 15213, USA (email: \href{mailto:abajcsy@cmu.edu}{abajcsy@cmu.edu}). Corresponding author: Jingqi Li.}
\thanks{This work was supported by DARPA under the Assured Autonomy (grant
FA8750-18-C-0101) and ANSR programs (grant FA8750-23-C-0080), the
NASA ULI program in Safe Aviation Autonomy (grant 62508787-176172),
and the ONR Basic Research Challenge in Multibody Control Systems (grant
N00014-18-1-2214). This work was also supported by the U.S. Army Research
Laboratory and Research Office (grant W911NF2010219).
}
}

\markboth{Journal of \LaTeX\ Class Files,~Vol.~14, No.~8, August~2021}%
{Shell \MakeLowercase{\textit{et al.}}: A Sample Article Using IEEEtran.cls for IEEE Journals}


\maketitle
\begin{abstract}%
    Autonomous agents should coordinate effectively without prior knowledge of others’ intents. While prior work has focused on intent inference, we address the \emph{inverse} problem: how agents can \emph{strategically} demonstrate their intents within general-sum dynamic games. {\finalrevise We model this problem and propose an algorithm that balances intent demonstration with task performance. To handle nonlinear dynamic games with continuous state–action spaces, our method leverages iterative linear–quadratic game approximations and provides efficient intent-teaching guarantees: the uncertain agent’s belief can be driven rapidly to the ground truth, while the demonstrating agent avoids expending effort on unnecessary belief alignment when it does not improve task performance. Theoretical analysis and hardware experiments confirm that our approach enables the demonstrating agent to reconcile task execution with belief alignment and strategically manage the information asymmetry among agents, even as its intent evolves during deployment.}
\end{abstract}


\section{Introduction}


{\revise{}General-sum dynamic games---wherein agents may have competing (but not opposing) objectives---are a powerful mathematical framework that can model a range of multi-agent behaviors, such as autonomous vehicle coordination \cite{schwarting2019social} and human-robot interaction \cite{music2020haptic}. 
When these models are put into practice
, an outstanding challenge is accounting for the fact that all agents' objectives (i.e., intents) may not be known \textit{a priori}. 
For example, when a car is merging onto the highway, the highway drivers typically pay attention to see if the new car is aggressively merging in front of them, or passively yielding to them.}

{\finalrevise Prior game-theoretic planners largely address intent uncertainty from the perspective of agents uncertain about others’ behavior, which we call \textit{uncertain agents}. These works typically assume the uncertain agent either acts under point estimates of others’ intents \cite{schwarting2019social,mehr2023maximum} or plans in expectation over a distribution of opponent strategies (e.g., aggressive vs. passive merging drivers) \cite{laine2021multi,le2021lucidgames}. Other approaches let the uncertain agent take information-gathering actions to probe the opponent’s intent \cite{sadigh2016information,hu2022active,yu2023active}, improving long-term performance. However, these models overlook the complementary perspective: the agent with certainty, referred to as the \textit{certain agent}, can also \textit{demonstrate} its intent. For instance, a merging driver may accelerate more aggressively to signal intent to surrounding vehicles. Our key insight is that \emph{a certain agent can intentionally shape uncertain agents’ beliefs through its actions, strategically managing information asymmetry to enhance overall task performance.}}

{\revise{}
In this work, we study \emph{strategic} intent demonstration in dynamic games, where a certain agent interacts with multiple uncertain agents. Our core idea is to model the certain agent as planning over both the evolution of the joint physical state and the dynamics of the uncertain agents’ beliefs. With this, we can design objectives enabling the certain agent to trade off between \emph{demonstrating its intent} (i.e., aligning the uncertain agents’ beliefs with its true intent) and \emph{pursuing its own task performance}, while the uncertain agents respond through belief updates and rational physical actions.

{\finalrevise Our primary contribution is a scalable continuous state-action algorithm for solving nonlinear intent demonstration games by iteratively approximating them with local linear-quadratic (LQ) games, i.e., games with linear dynamics and quadratic objectives.} Our algorithm consists of two sub-optimizations: first solving for all agents' game-theoretic feedback policies parameterized by \textit{any} intent, and then solving the certain agent's optimization over the joint physical and estimate dynamics. We theoretically characterize the convergence of the uncertain agents' beliefs and the certain agent’s ability to balance intent demonstration with task performance. We also evaluate our method in a suite of multi-agent settings such as decentralized bi-manual robot manipulation, three-vehicle platooning, and shared control. We find that when agents can strategically demonstrate their (dynamically changing) intents to others, they can achieve superior task performance and coordination.







\section{Related Works}\vspace{-0.5em}


{\revise{}\para{Efficient Solutions to General-Sum Dynamic Games} Even without intent uncertainty, solving general-sum dynamic games over continuous state and action spaces is challenging. 
Most of these games have no analytic solution, and classical dynamic programming approach for finding Nash equilibria of these games suffers from the ``curse of dimensionality'' \cite{powell2007approximate}. 
However, under linear dynamics and quadratic costs, there exist efficient numerical solutions for solving these linear-quadratic (LQ) games \cite{bacsar1998dynamic}. Recent works propose to solve nonlinear games by iteratively approximating them via LQ games \cite{fridovich2020efficient}. 
In this work, we leverage these fast and approximate iterative LQ game solvers as a submodule in our intent demonstration algorithm.}

\para{Incomplete Information Games} 
Prior dynamic programming solutions to incomplete information games \cite{harsanyi1968games,ouyang2016dynamic, vasal2018systematic, huang2019dynamic, sagheb2023should} do not scale to high-dimensional nonlinear games with continuous state, action and intent spaces. Thus, 
recent works focus on scalable approximations. 
One overarching approximation is assuming that some agents have complete information and others do inference. These approaches model the uncertain agents as planning in expectation \cite{schwarting2021stochastic, laine2021multi}, planning with the most likely estimate and recovering a complete-information game \cite{le2021lucidgames,chahine2023intention}, doing intent inference from an offline dataset \cite{peters2021inferring,li2023cost, mehr2023maximum}, planning multiple contingencies based on discrete intent hypotheses \cite{peters2023contingency}, and modeling incentives for uncertain agents to take information-gathering actions \cite{sadigh2016information, hu2022active}. While prior works focus on how \textit{uncertain agents} should tractably plan under their beliefs, we focus on how the \textit{certain agent} can demonstrate their intent by exploiting the learning dynamics of other agents. 

{\revise{}\para{Intent Demonstration in Multi-Agent Interactions} Prior works on intent demonstration, such as legibility in robot motion planning around humans \cite{dragan2013legibility}, typically model uncertain agents as passive observers. However, in scenarios like multi-agent highway driving \cite{xing2019driver} or collaborative manipulation \cite{losey2018review}, all agents actively interact while some simultaneously learn the missing information of the games. Unlike previous multi-agent intent demonstration frameworks \cite{mavrogiannis2018social,bastarache2023legible}, our model explicitly accounts for rational feedback from uncertain agents within general-sum dynamic games. Moreover, rather than simply aligning uncertain agents’ beliefs with the certain agent’s true intent, our approach allows the certain agent to strategically shape these beliefs, thereby guiding uncertain agents’ actions to enhance overall task~performance beyond conventional~belief-alignment methods~\cite{dragan2013legibility,le2021lucidgames}. }

\section{Background: General-sum Dynamic Games}
\label{sec:background}
In this section, we present the necessary background on general-sum dynamic games. 

\para{Notation} We consider general-sum games played over the finite time horizon $T$.
We consider $N$ players in the game, each of whose control action is denoted by 
$u_t^i \in \mathbb{R}^m$ for $i\in\{1, 2, \dots,N\}$. 
Let the set of times $\{0,1,\hdots,T\}$ be denoted by $\T$ and the set of player indices $\{1, 2, \dots,N\}$ be denoted by $\N$. 
We denote $x_t\in\mathbb{R}^n$ to be the joint physical states of all players (e.g., positions, velocities) which evolves via the deterministic discrete-time  dynamics, $x_{t+1}=f_t(x_t,u_t^1,\dots,u_t^N)~\forall t\in \T,$
where $f_t(\cdot): \mathbb{R}^n\times \mathbb{R}^{m}\times \dots\times \mathbb{R}^m\to \mathbb{R}^n$ is assumed to be a differentiable function.  
For notational convenience, we denote the vector of all $N$ agents' actions at time $t$ to be $u_t:=[u_t^1,\dots,u_t^N]$. 

\para{Player Objectives} Let each player $i\in\N$ seek to minimize their own cost function, $c_t^i(x_t, u_t)$. Note that in general this cost function depends on \textit{both} the joint physical states of all players and also the actions of all players.
It is precisely this coupling that induces a dynamic game between all players. 
The Nash equilibrium defines a scenario wherein no player wants to deviate from their current state-action profile under their respective cost functions. 
Specifically, in our work, we consider \textit{feedback Nash equilibrium} (FNE) \cite{bacsar1998dynamic}, wherein each player $i\in\N$ solves for a policy $\pi_t^i(x_t):\mathbb{R}^n \to \mathbb{R}^m$ which gets access to the current joint physical state, $x_t$, at any time, and outputs an action.
When the cost functions for all agents were assumed to be known a priori, such games are called \emph{complete information games}. 
However, when players have uncertainty over other players' objectives, these are \emph{incomplete information games}, which is what we study here.

\section{Problem Formulation}

We study the problem of intent demonstration---wherein one agent can express their intent to uncertain agents---in general-sum game-theoretic interactions. 
Similar to prior work \cite{chahine2023intention,le2021lucidgames}, we consider an information \textit{asymmetry} between the players: one player (e.g., player 1) has complete information, i.e., they know the cost functions of all players, but players 2 through $N$ have incomplete information  about player 1's cost function. 

Moreover, we assume that each agent is aware of its status as either certain or uncertain, and that this information is shared among all agents. 
For example, from our introductory example, the driver merging in from an on-ramp has certainty over their own driving style, but all other road agents on the highway do not. 
However, players 2 through $N$ have the ability to \textit{estimate} or learn about player 1's cost function during game-theoretic interaction. This problem cannot be reformulated as another complete information dynamic game with deterministic dynamics because players 2 through N are not aware of player 1’s cost function and there can be an infinite number of possible cost functions for player 1. We formalize these ideas below. \vspace{-0.25em}
 
\para{Certain Player: Cost Parameterization}
Without loss of generality, let player 1 be the agent with complete information of the game, including the cost functions of other players.
We model player 1's task-centric cost function, $c^1_t(x_t, u_t; \trueparam)$, as parameterized by a low-dimensional parameter, $\trueparam \in \Theta$, which could in theory be discrete (e.g., aggressive or passive driving style) or continuous (e.g., weights on a linear feature basis). In what follows, we consider the more challenging setting of continuous intent parameters with infinitely many possible values. \vspace{-0.25em} 

\para{Uncertain Players: Estimation and Cost Functions}
All agents, except for player 1, are uncertain about player 1's cost function parameter. They maintain \textit{estimates} of this parameter via $\est$, which in general can be a full Bayesian belief or a point estimate. 
All uncertain agents possess the ability to learn, based on the  joint physical states ($x_t$) and the action of player 1 ($u_t^1$) observed during interaction. 
For any uncertain player $j\in\{2,3\dots,N\}$ and their associated estimate $\est^j_t$ at time $t$, let $\est^j_{t+1} = g_t(\est_t^j, x_t, u_t^1)$ be the updated estimate via update rule $g_t$. 
Each uncertain player aims to minimize their own cost function $c_t^j(x_t,u_t)$.\vspace{-0.25em}

\para{Intent Demonstration Formulation}
We can now formulate the intent demonstration problem in general-sum games. One of our core ideas is to augment player 1's state space with the estimates of all uncertain agents' beliefs. 
Let the vector of all uncertain agents' current estimates be denoted by $\est_t := [\est_t^2,\dots,\est_t^N]$.
We model the certain agent's cost as a combination of their task-centric cost, $c_t^1(x_t,u_t;\trueparam)$, (e.g., for an autonomous car this could be lane-keeping and smoothness of motion), and the ``error'' between the uncertain agent's estimates and the true intent, $\DemoCost(\est_t,\trueparam)$, (e.g., expressing that they are aggressive or in a rush): \vspace{-0.2em}
\begin{equation*}
    \TeachingCost_t^1(x_t,\est_t,u_t;\trueparam):= \taskweight\cdot c_t^1(x_t,u_t;\trueparam)+\intentweight \cdot \DemoCost(\est_t, \trueparam),\vspace{-0.2em}
\end{equation*}
where $\taskweight,\intentweight\ge 0$ are hyper-parameters. 
Intuitively, this enables player 1 to synthesize a range of behaviors, from prioritizing task-cost and only influencing the uncertain agents' beliefs when beneficial for minimizing task cost (i.e., $\taskweight > 0, \intentweight \equiv 0$), to encouraging player 1 to actively express their intent (i.e., $\taskweight \equiv 0, \intentweight > 0$).  When the task cost $c_t^1$ includes safety-related terms, the weighted cost $\bar{c}_t^1$ balances task performance and intent demonstration, thereby enabling efficient intent demonstration without compromising safety.

Ultimately, player 1's intent demonstration problem 
optimizes their augmented cost subject to several key constraints:\vspace{-0.5em}
\begin{subequations}\label{eq:active teaching problem}
    \begin{align}
        \min_{\{u_t^1\}_{t=0}^T} &\sum_{t=0}^T \TeachingCost_t^1(x_t,\est_t,u_t;\trueparam) \\
        \textrm{s.t. }& x_{t+1} = f_t(x_t,u_t),\forall t \in \T  
        \label{eq:phys_dyns}\\
        & \est_{t+1}^j =g_t(\est_t^j,x_t,u_t^1),\forall j\in \N\setminus \{1\}, \forall t \in \T 
        \label{eq:est_dyns}\\
        & u_t^j = \pi_t^j(x_t;\est_t^j),\forall j\in \N\setminus \{1\},\forall t \in \T 
        \label{eq:uncertain_policy}\\
        & x_0 = x^\mathrm{init}
        , ~\est_0 = \est_{\mathrm{init}} \label{eq:init_cond}.  
    \end{align}
\end{subequations}
Here, Equations~\eqref{eq:phys_dyns} and \eqref{eq:est_dyns} constrain the solution to abide by the  physical dynamics of the joint system and ensure that the estimates of the uncertain players follow their update rules. 
Given any player's current estimate $\est_t^i$, Equation~\eqref{eq:uncertain_policy} models the uncertain players as rationally responding under their current FNE strategy\footnote{\color{black}We assume that there exists a unique FNE. When multiple FNEs exist, we can use the technique in \cite{peters2020inference} to ensure agents pursuing a consistent~FNE.  
} 
$\pi_t^i(x_t;\est_t^i)$, \textit{assuming that all agents also play under player $i$'s current intent estimate, $\est_t^i$}.  
Note that this is simply a virtual game model in the mind of each uncertain player (see purple dashed box in Figure~\ref{fig:decision process}). In reality, player 1 can behave differently than the current estimate $\est_t^i$; however, this is not a problem for uncertain players since they will update their intent estimate at the next timestep. 
Finally, similar to prior first-order belief assumptions \cite{schwarting2021stochastic}, in Equation~\eqref{eq:init_cond} we assume that the initial estimates, $\est^j_0$, of each uncertain player $j\in\{2,\dots,N\}$ are common knowledge. 
An illustrative diagram of our interaction model between two players is visualized in Figure~\ref{fig:decision process}. 
\begin{figure*}[t!]
    \centering
    \includegraphics[width = 0.95\linewidth]{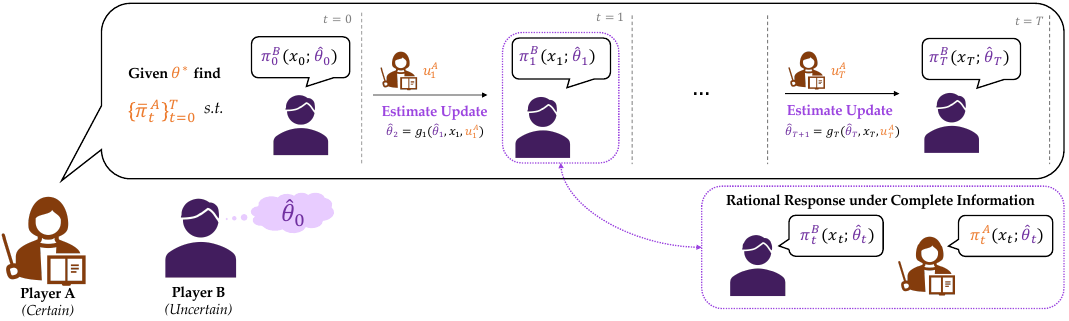}
    \vspace{-1.1em}
    \caption{\textbf{Intent Demonstration Problem in General-Sum Games.} The \textit{certain} player A optimizes $u_t^A=\TeachingPolicyUnLabeled_t^A(x_t, \est_t ;\trueparam)$, which trades off its own task cost and demonstrating their intent. The \textit{uncertain} player B engages with player A through rational actions $u_t^B =\policy_t^B(x_t;\est_t)$ and updates their estimate $\est_t$ of player $A$'s intent $\trueparam$ by observing $A$'s actions. This enables player $A$ to choose to influence player $B$'s estimate.}
    \label{fig:decision process}\vspace{-1.2em}
\end{figure*}

{\revise{}\section{Theoretical and Algorithmic Results}\vspace{-0.3em}
In this section, we investigate both the theoretical and algorithmic properties of the proposed intent demonstration formulation, involving a certain player (player 1) and uncertain players (players 2 to N). 
{\finalrevise We begin by analyzing LQ games, a class of games well-suited for theoretical analysis, as the dynamics $f_t$ are linear and each player’s cost $c_t^i$ is quadratic for all $t \in \T$ and $i \in \N$.} We establish rigorous theoretical guarantees concerning the efficiency of intent teaching. Subsequently, we extend our framework algorithmically to address intent demonstration problems within multi-player nonlinear games, such as those incorporating nonlinear Bayesian estimation rules.\vspace{-0.3em}

\subsection{LQ Games with Linear Estimation Dynamics}\label{subsection:LQ games}\vspace{-0.3em}

\para{LQ Setup} 
We consider settings where player 1's true intent parameter is a continuous goal parameter (i.e., part of their cost). Each player $j\in\{2,\dots,N\}$ maintains a point estimate $\est_t^j$ of $\trueparam$. Let the joint physical dynamics in optimization problem \eqref{eq:active teaching problem} be a time-varying linear system,
$
f_t := A_tx_t + \sum_{i=1}^N B_t^i u_t^i, \quad t \in \T, 
$
with $A_t\in\mathbb{R}^{n\times n}$ and $B_t^i \in \mathbb{R}^{n\times m}$. Let player 1's task and intent-demonstration costs be quadratic in physical state and control: $c_t^1(x_t,u_t;\trueparam)  := x_t^\top Q_t^1 x_t + u_t^{1\top }R_t^1 u_t^1 + x_t^\top \trueparam$ and  $\DemoCost(\est_t, \trueparam) := \sum_{j=2}^N \|\est_t^j - \trueparam\|_2^2$. Similarly, for each $i\in \{2,\dots,N\}$, let player $i$'s quadratic cost be $c_t^i(x_t,u_t)  := x_t^\top Q_t^i x_t + u_t^{i\top }R_t^i u_t^i$
where $Q_t^i\in\mathbb{R}^{n\times n}$ and $R_t^i\in\mathbb{R}^{m\times m}$ are positive semi-definite matrices. 

\para{Uncertain Player's Feedback Policy} 
Given their current point estimate, $\est_t^j$, the uncertain player $j$ rationally responds under their current FNE policy $\pi_t^j(x_t;\est_t^j)$, assuming a complete information game where player $1$ also acts rationally under player $j$'s estimate, $u_t^1=\pi_t^1(x_t;\est_t^j)$. 
Importantly, since we are in the LQ setting, all players' complete information game FNE policies are linear feedback policies \cite{bacsar1998dynamic}. 
%

\para{Linear Estimation Dynamics} Finally, let the estimate dynamics of an uncertain player $j\in\{2,\dots,N\}$ be linear in state and estimate. Considering a step size $\stepsize>0$, we study a gradient descent-based maximum likelihood estimation (MLE) update rule \cite{losey2019learning}, $g_t(\est_t^j,x_t,u_t^1)$, as \vspace{-0.3em}
\begin{equation}
    g_t := \est_t^j - \stepsize \nabla_{\est_t^j}\Big(\frac{1}{2}\|u_t^1 - \pi_t^1(x_t;\est_t^j)\|_2^2\Big).
    \label{eq:mle_linear_update_rule}\vspace{-0.3em}
\end{equation}
Player $j$ updates their estimate based on the difference between the action they expected player 1 to take under their estimate, $\pi_t^1(x_t;\est_t^j)$, and player 1's observed action,~$u_t^1$.  

\para{Bellman Equation and Algorithm} When an uncertain player learns via a linear MLE update rule, intent demonstration is an LQR problem in the joint physical state $x_t$, the estimate $\est_t$, and the true cost parameter $\param^*$. The Bellman equation for player 1's intent demonstration problem specified in Equation~\eqref{eq:active teaching problem}  is defined as:\vspace{-0.3em}
\begin{equation}
\begin{aligned}
V_t^1(x_t,\est_t;\trueparam) =&\min_{u_t^1}  \TeachingCost_t^1(x_t,\est_t,u_t^1,\{\pi_t^j(x_t;\est_t^j)\}_{j=2}^N;\trueparam ) \\& + V_{t+1}^1(x_{t+1},\est_{t+1};\trueparam).  \vspace{-0.3em} 
\end{aligned}
\label{eq:value_function}
\end{equation}
With this Bellman equation in hand, we can now pose our intent demonstration Algorithm~\ref{alg:LQ games} and leverage a suite of off-the-shelf numerical techniques for each component of our algorithm. Specifically, in Algorithm~\ref{alg:LQ games}, we first solve a complete information linear-quadratic game for all players under each possible intent parameter $\param \in \Theta$. Importantly, here we can obtain feedback policies, $\{\pi_t^i\}_{i=1,t=0}^{N,T}$, for all agents with efficient (polynomial time) off-the-shelf algorithms. These feedback policies are re-used by all players. Player $j\in\{2,\dots,N\}$ uses $\pi_t^1(x_t;\est_t^j)$ to predict player 1's actions under their current estimate, $\est_t^j$, and then updates the estimate. 
Player $1$ plans over the estimation dynamics of players $j\in\{2,\dots,N\}$ when it solves the LQR problem leveraging the value function specified in Equation~\eqref{eq:value_function}. Once again, this yields a feedback control law for player 1 in the joint physical and estimate state space, $\TeachingPolicy(x_t, \est_t; \trueparam)$, and enjoys the benefits of off-the-shelf LQR solvers. We note that the active intent demonstration policy computed by Algorithm~\ref{alg:LQ games} is guaranteed to converge to the optimal one when the associated LQ games and the LQR problems are well-defined and admit valid solutions. 


\begin{algorithm}[t!]
\caption{Strategic Intent Demonstration Games}\label{alg:LQ games}
\begin{algorithmic}[1]\revise{}
\REQUIRE dynamics $f $, player 1's task cost $c_t^1(x, u; \trueparam)$ and demonstration cost $\DemoCost(\est, \trueparam)$, $\taskweight,\intentweight\ge0$, player $j$'s cost $c_t^j(x, u)$, initial estimate $\est_0^j$, for each $j\in\{2,\dots,N\}$, and 
estimation~dynamics~$g$ \\
\tcp{Solve complete information policies}
\STATE $\{{\pi_t^i(x;\param)}\}_{i=1,t=0}^{N,T} \gets \texttt{FeedbackGame}(\{c_t^i\}_{i=1}^N, f)$ \label{algstep: solving games}
\STATE $\Pi = \{{\pi_t^i(x; \theta)}\}_{i=1,t=0}^{N,T}$\\
\tcp{Compute intent demonstration policy}
\STATE $\{\TeachingPolicy(x, \est; \trueparam)\}_{t=0}^T \gets \texttt{OptimalControl}(\est_0,\TeachingCost^1_t, f, g)$ \label{algstep: solving teaching}
\STATE $\Pi \gets \{\TeachingPolicy(x,\est; \trueparam)\}_{t=0}^T \cup \Pi$
\RETURN $\Pi$
\end{algorithmic}
\end{algorithm}

\para{Theoretical Results} Finally, in the LQ setting, we prove a sufficient condition for the existence of an intent demonstration policy for player 1 which guarantees to drive player $j$'s estimate to the true parameter exponentially fast, for all $j\in\{2,\dots,N\}$. Our proof operates under player 1's cost, $\TeachingCost^1_t$, with $\taskweight =0 $ and $\intentweight>0$, meaning that player $1$ only considers demonstrating their intent.  
\begin{proposition}[Effective Intent Demonstration]\label{prop:teaching guarantee}
    Consider an N-player LQ game. Suppose that the linear policy $\pi_t^1(x_t;\param)$ takes the form
    $
        \pi_t^1(x_t;\param) =  K_{t,x}^1 x_t  + K_{t,\param}^1 \param
    ,\ \forall t\in \T
    $
    and $K_{t,\param}^{1\top} K_{t,\param}^1>0$. Moreover, let player $j\in\{2,\dots,N\}$ learn via linear estimate dynamics $\est_{t+1}^j = g_t(\est_t^j, x_t, u_t^1) $.
    Pick a step size $\stepsize\in(0,1)$ such that the largest singular value of $(I-\stepsize K_{t,\param}^{1\top}K_{t,\param}^1)$ is less than $1$, $\forall t\le T$. Then, there exists a linear intent demonstration policy $u_t^1=\bar{\pi}_t^1 (x_t,\est_t; \trueparam)$ such that $\|\est_{t+1}^j - \trueparam\|_2 \le c \|\est_t^j - \trueparam\|_2 $, $\forall t\in\T$, $\forall j \in \{2,\dots,N\}$,
    where $0<c<1$ is a constant dependent on $\bar{\pi}_t^1$. 
\end{proposition}

\begin{proof}
    The proof can be found in the Appendix.
\end{proof}

Proposition~\ref{prop:teaching guarantee} is a feasibility result, and the strong assumption on the form of the policy $\pi_t^1$ is not necessary for the existence of active intent demonstration policies. Moreover, always actively demonstrating the intent to other uncertain agents could be excessive and may impair the certain agent's task performance. We show in the following result that the active teaching policy can trade-off between the certain agent's task completion and intent demonstration such that it can achieve a task performance even higher than in the complete information game, when setting $\taskweight>0$ and $\intentweight=0$.
\begin{proposition}[Strategic Intent Demonstration]\label{prop:strategic teaching}
    Let $\taskweight=1$ and $\intentweight=0$. 
    Define $g_t$ as in \eqref{eq:mle_linear_update_rule}. Suppose that each player's cost is convex with respect to $(x_t,u_t)$. Let $\{\tilde{u}_t^{i}\}_{i=1,t=0}^{N,T}$ be the controls corresponding to the Nash equilibrium in the complete information game, and denote by $\{\tilde{x}_t\}_{t=0}^{T+1}$ the resulting state trajectory. Suppose $\exists t <T$ such that the gradient $\nabla_{u^1_{t:T}}\tilde{c}_{t:T}^1$, where $u^1_{t:T}:= [u^1_\tau]_{\tau=t}^T$ and $\tilde{c}_{t:T}^1$ is defined below, is nonzero at $\tilde{u}^1_{t:T}$,\vspace{-0.4em}
\begin{equation*}
    \begin{aligned}
        \tilde{c}_{t:T}^1(\tilde{x}_t, u_{t:T}^1)\hspace{-0.15em}:=  \hspace{-0.15em} \sum&_{ \tau = t }^T c_\tau^1(x_\tau, u_\tau^1, \{\pi_\tau^j(x_\tau; \est_\tau^j)\}_{j=2}^N; \trueparam)\\
        \textrm{s.t. }& x_{\tau+1} = f_\tau(x_\tau,u_\tau^1, \{\pi_\tau^j(x_\tau; \est_\tau^j)\}_{j=2}^N), \\
        & \est_{\tau+1}^j = g_\tau(\est_\tau^j, x_\tau, u_\tau^1),  j\in \N\setminus\{1\}\\
        & \tau\in[t, T], x_t = \tilde{x}_t, \est_t^j = \trueparam, j\in \N\setminus\{1\}
    \end{aligned}
    \end{equation*}
    then, player 1's optimal cost in \eqref{eq:active teaching problem} is strictly less than its optimal cost in the complete information game ($\est_0=\trueparam$). 
\end{proposition}
\begin{proof}
    The proof can be found in the Appendix. 
\end{proof}





\begin{figure*}[t!]
    \centering
    \includegraphics[width=0.7\linewidth]{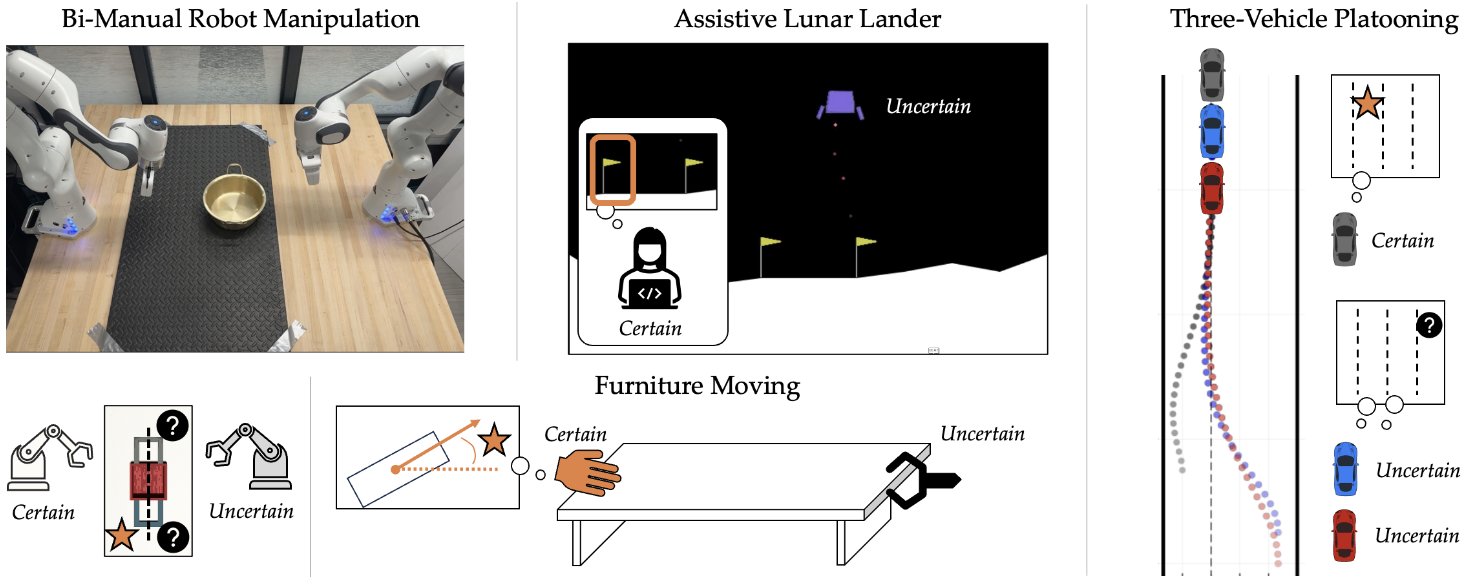}
    \vspace{-1em}
    \caption{\textbf{Environments.} Four incomplete information general-sum games considered in this work. }
    \label{fig:environments}
    \vspace{-1.5em}
\end{figure*}

Proposition~\ref{prop:strategic teaching} suggests that the ability of influencing the uncertain agent's belief enables the certain agent to achieve a higher task performance. 
In practice, we can replace the estimation dynamics in \eqref{eq:mle_linear_update_rule} with other types of estimation dynamics, e.g., Bayesian inference or general maximum likelihood estimation. We will explore this in the next subsection. 

\subsection{Nonlinear Games with Nonlinear Estimation Dynamics}\vspace{-0.1em}

With small modifications, we can adapt Algorithm \ref{alg:LQ games} to non-quadratic costs and nonlinear dynamics. This is particularly important as many estimation update rules, such as the Bayesian belief update, are nonlinear in the estimate.

{\finalrevise Our algorithm builds on iterative LQ games (\texttt{iLQGames}) and iterative LQR (\texttt{iLQR}), which have polynomial computational complexity in the system dimension and have been shown to handle high-dimensional problems efficiently \cite{fridovich2020efficient,todorov2005generalized}. Leveraging this structure, our method remains tractable and scales to multi-agent settings. Similar to Algorithm~\ref{alg:LQ games}, we approximate the complete-information FNE policies using an \texttt{iLQGames} solver \cite{fridovich2020efficient}, which iteratively linearizes the dynamics and quadratically approximates the costs around the current trajectory. 
Such approximations provide sufficient local curvature information to compute approximate Nash equilibria that guide the nonlinear system toward equilibrium. The solver then computes an optimal control for the local game, updates the trajectory, and repeats until convergence.}

Similar to Section~\ref{subsection:LQ games}, after we compute the complete information \texttt{iLQGames} policies $\{\pi_t^i(x_t;\param)\}_{i=1}^N$, we use these policies, once again, in both the uncertain player's estimation dynamics and for the certain player's intent demonstration. 
For example, if player $j\in\{2,\dots,N\}$ maintains a Gaussian belief $\est_t^j := b_t^j(\param)=\mathcal{N}(\mu_t^j, \Sigma_t^j)$ over the intent parameter and learns via a nonlinear belief update rule like Bayesian inference, they use $\pi^1_t$ computed from \texttt{iLQGames} to construct their (Gaussian) likelihood function and obtain the posterior $b_{t+1}^j(\param) \propto p(u_t^1| x_t, \param) \cdot b_t^j(\param)$. 
Assuming that the likelihood model $p(u_t^1| x_t, \param)$ follows a Gaussian distribution $ \mathcal{N}(\pi_t^1(x_t;\param),I)$, and the initial belief is also a Gaussian distribution $\mathcal{N}(\mu_0^j, \Sigma_0^j)$, we can simplify the belief update by substituting the policy $\pi_t^1(x_t;\param) $ and obtain the extended Kalman filter update rule:\vspace{-0.4em}
\begin{equation}
    \begin{aligned}
        \mu_{t+1}^j = & \mu_t^j + \Sigma_t^j \cdot \nabla_{\param} \pi_t^{1\top} \cdot (I + \nabla_{\param} \pi_t^1 \cdot \Sigma_t^j \cdot \nabla_\param \pi_t^{1\top})^{-1} \cdot \\
        &(u_t^1 - \pi_t^1(x_t; \mu_t^j)) \\
        \Sigma_{t+1}^j = & \Sigma_t^j - \Sigma_t^j \cdot \nabla_\param \pi_t^{1\top} \cdot (I + \nabla_\param \pi_t^1 \cdot \Sigma_t^j\cdot \nabla_\param \pi_t^{1\top})^{-1} \cdot \\ & \nabla_\param \pi_t^1 \cdot \Sigma_t^j
    \end{aligned}\label{eq:ekf}
\end{equation}
To optimize $c^{demo}(\cdot, \cdot)$, the certain agent can, for example, minimize the error between the average intent under the other agent's belief, $\tilde{\param}_t^j := \mathbb{E}_{\param \sim b_t^j(\param)} [\param]$, and $\trueparam$. From player 1's perspective, instead of solving an LQR problem as in Section~\ref{subsection:LQ games}, it solves an \texttt{iLQR} problem to obtain the intent demonstration policy $\bar{\pi}_t^1$ in the joint physical-estimate space.

\begin{remark}
    We can enhance Algorithm~\ref{alg:LQ games} by integrating deep reinforcement learning to compute policies for intent demonstration problems in general-sum nonlinear dynamic games. For instance, multi-agent reinforcement learning \cite{lowe2017multi} can be applied in step \ref{algstep: solving games} of Algorithm~\ref{alg:LQ games} to compute complete-information FNE policies, while deep reinforcement learning can be used in step \ref{algstep: solving teaching} of Algorithm~\ref{alg:LQ games} to derive a strategic intent demonstration policy.
\end{remark}
}

\section{Experiments}\label{sec:experiment}

In this section, we evaluate our algorithm\footnote{The source code and additional details of the experiments are available at \href{https://github.com/jamesjingqili/Active-Intent-Demonstration-in-Games.git}{https://github.com/jamesjingqili/Active-Intent-Demonstration-in-Games.git}.} in four multi-agent scenarios shown in Figure~\ref{fig:environments}
and study the benefits of strategic intent demonstration over alternative game-theoretic interaction models.

\para{Bi-Manual Robot Manipulation}
{\finalrevise We study a robot manipulation task where two arms must coordinate to lift a pot (top left, Figure~\ref{fig:environments}). The certain agent (left) aims to grasp one handle, while the uncertain agent (right) does not know this intent. The left agent’s preferred $y$-goal is $\trueparam \in \mathbb{R}$ (lower handle), and the right agent maintains an evolving estimate $\est \in \mathbb{R}$ via the update rule in Equation~\eqref{eq:mle_linear_update_rule}. The joint state is $x_t = [p_{x,t}^1,p_{y,t}^1,p_{x,t}^2,p_{y,t}^2]$, where players control end-effector velocities $u_t^i=[v_{x,t}^i,v_{y,t}^i]$, $i\in\{1,2\}$, under single-integrator dynamics. The left agent’s quadratic task cost penalizes distance to the target, collision, and high velocity, while the right agent optimizes a similar cost but is incentivized to grasp the opposite handle. We validated our method in hardware experiments.}

\para{Assistive Lunar Lander}
A lunar lander autopilot shares control with a human pilot. 
The human pilot controls horizontal thrust and wants to land at their preferred destination on the x-axis (top center Figure~\ref{fig:environments}), $\trueparam \in \mathbb{R}$, which is unknown to the autopilot.
The autopilot controls both the vertical and horizontal thrust, aiming to avoid crashing on the ground while conserving fuel. For the convenience of analysis, we focus on its horizontal and vertical movements, excluding the rotation dynamics, and model this interaction as a two-player 
linear-quadratic game. The autopilot maintains a point estimate $\est \in \mathbb{R}$ and learns via linear estimate update rule (e.g., as in Equation~\eqref{eq:mle_linear_update_rule}).

\para{Furniture Moving}
A human and robot must move a table to a known destination together. The human's task cost is parameterized by their desired furniture moving angle, $\trueparam$, and they seek to minimize their effort. The robot maintains a Bayesian belief $\est := b(\param)$ over the human's preferred orientation angle (bottom, Figure~\ref{fig:environments}). 
The joint physical state is position and current table angle $x_t=[p_{t,x}^H,p_{t,y}^H, p_{t,x}^R,p_{t,y}^R, \param_t]$ and players control their $x$ and $y$ velocity. The dynamics of the furniture moving follows a simple kinematics model. The robot learns via a Bayesian belief update.

\para{Three-Vehicle Platooning}
A human driver guides two autonomous vehicles (AV) towards a target lane, unknown to the autonomous vehicles. 
Each vehicle has a unicycle dynamics with a state vector $x_t^i := [p_{t,x}^i,p_{t,y}^i,\psi_t^i, v_t^i]$ and control inputs are acceleration $a_t^i$ and turning rate $w_t^i$ (12-D joint state vector). 
{\revise{Each AV optimizes 1) following the human driver's lane, 2) maintaining a forward orientation, 3) minimizing control effort and 4) avoiding collisions. }}Each AV has a separate Gaussian belief over the human driver's target lane, $\est^i := b^i(\param)$, and updates via Bayesian estimation.

\begin{figure}[t!]
    \centering
    \includegraphics[width=\linewidth]{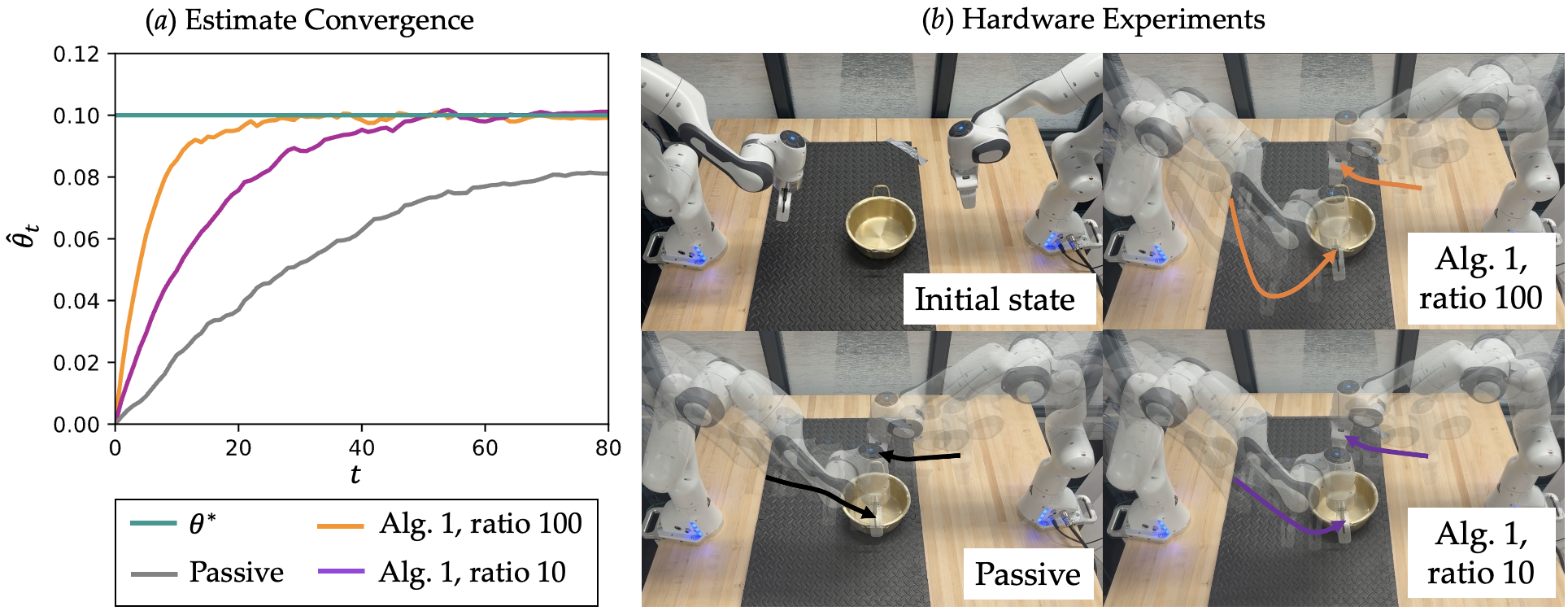}
    \vspace{-1em}
    \caption{\textbf{Results: H1.} \ours accelerates learning by having the certain agent exaggerate its behavior, helping the uncertain agent infer its intent.
    }
    \label{fig:manipulation}\vspace{-0.3em}
\end{figure}
\begin{figure}[t!]
    \centering
    \includegraphics[width=0.9\linewidth]{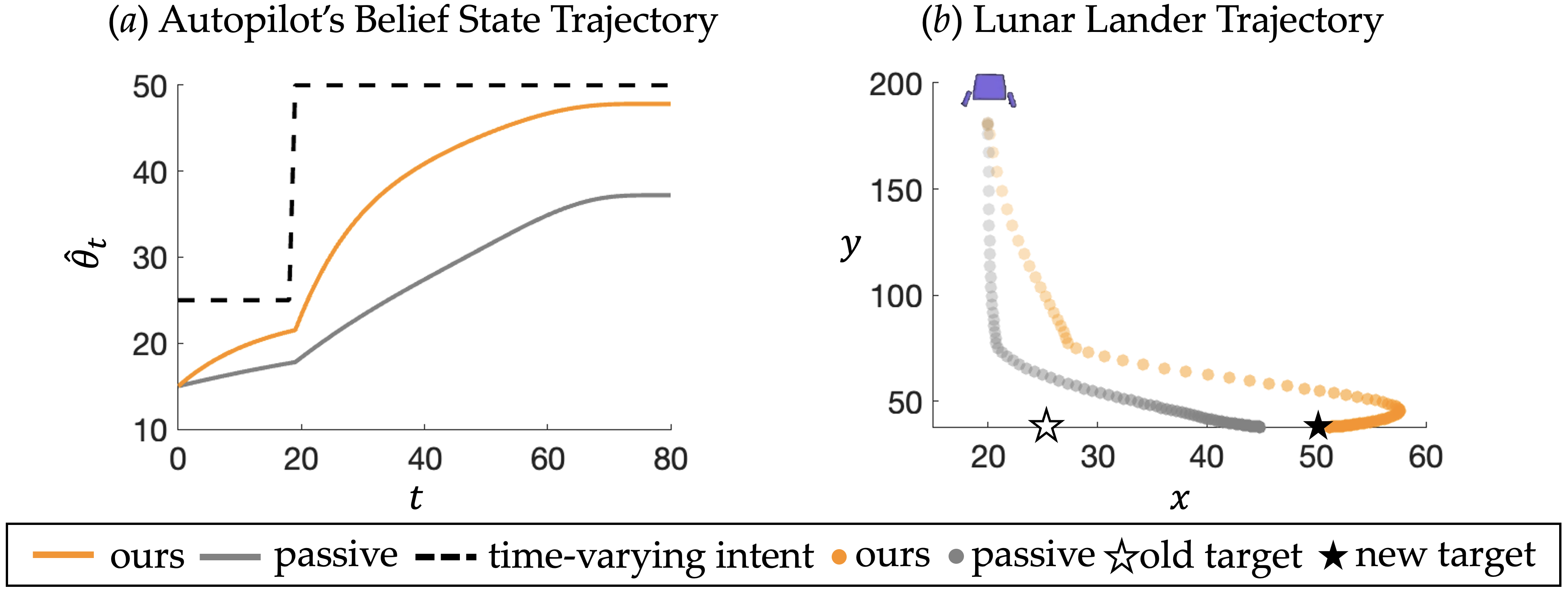}\vspace{-0.8em}
    \caption{\textbf{Results: H2.} \revise{}The human pilot changes their target landing position $\trueparam$ from 25 to 50 at time $t=20$. The strategic intent demonstration policy $\bar{\pi}_t^1$, computed without anticipating this change, efficiently conveys the unforeseen dynamic intent, enabling the autopilot’s belief to converge faster than in the passive game, without the need of recomputing $\bar{\pi}_t^1$.
    }
    \label{fig:varying}\vspace{-1.em}
\end{figure}
\subsection{Empirical Results}
\begin{figure}[t!]
    \centering
    \includegraphics[width=0.9\linewidth, ]{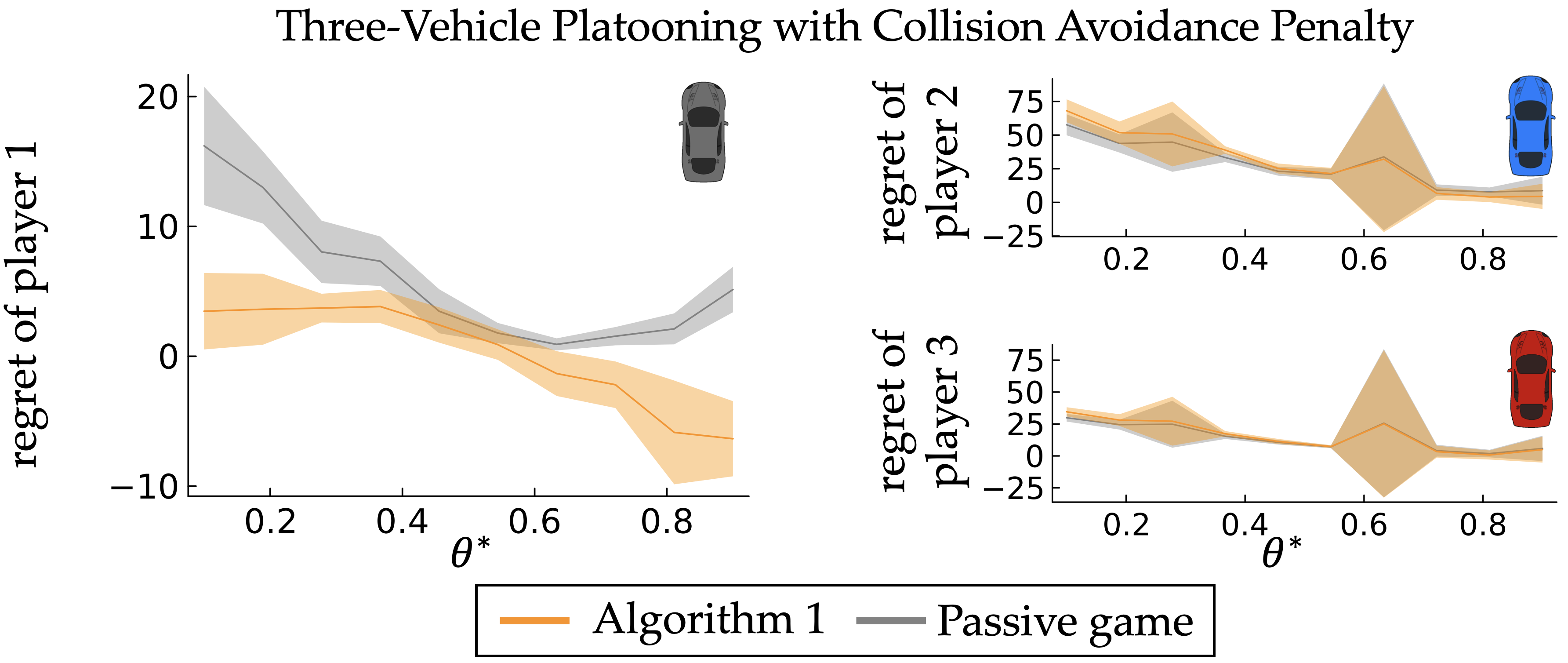}\vspace{-0.9em}
    \caption{\textbf{Results: H3.} The regrets of the certain player (player 1) under the active teaching strategy are consistently lower compared to those under the passive teaching strategy, across different ground truth intents of the certain player. This empirically validates the claim in Proposition~\ref{prop:strategic teaching}. 
    }
    \label{fig:regret}\vspace{-1.6em}
\end{figure}

\begin{figure*}[t!]
    \centering
    \includegraphics[width=\textwidth]{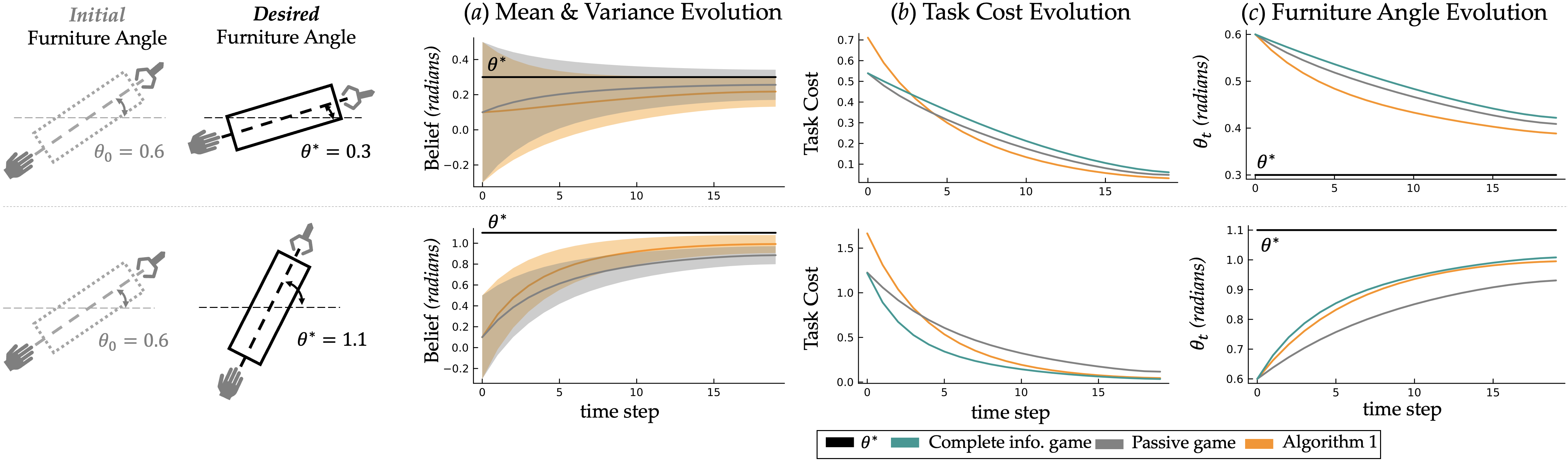}\vspace{-1.2em}
    \caption{\textbf{Results: H4}.
    Even without explicit incentives to express intent, \ours influences the uncertain agent's belief in a way that improves task cost over \passive game (plot (b)). However, intent demonstration is strategic: if the state is already sufficiently good, our method pauses its influence (top left, plot (a)) but still achieves better task performance. 
    }
    \label{fig:furniture}
    \vspace{-1.6em}
\end{figure*}
{\color{black}We compare \ours with two other models. The first is referred to as a \emph{\passive} game, built on the state-of-the-art incomplete information game solver \cite{le2021lucidgames}: the certain agent acts under the complete information FNE, while uncertain agents infer its intent from observed actions and then plan their actions, so their learning is not anticipated by the certain agent. When uncertain agents hold Gaussian beliefs, uncertain agents use an extended Kalman filter belief update \eqref{eq:ekf}. When they hold point estimates, where the Kalman filter does not apply, uncertain agents use the MLE update \eqref{eq:mle_linear_update_rule}. The second model is a \emph{\oracle} game in which agents know each other's intent and execute the complete information FNE. We study four hypotheses below.}

\noindent \textit{\textbf{H1.} Uncertain agents coordinating under \ours reduce uncertainty faster than \passive game-theoretic models that do not account for agent learning.} 

\para{Setup and Metrics} We focus on the \textbf{Bi-Manual Robot Manipulation} environment where the uncertain agent maintains a point estimate. The uncertain robot initially believes that the certain robot wants to grab the center of the pot, $\est_0 = 0$.
We measure the convergence of $\est_t$ to $\trueparam$ under \passive game and \ours. 
{\revise{}For our method, we also vary the hyperparameters $\taskweight$ and $\intentweight$, denoted by $\text{ratio} = \frac{\intentweight}{\taskweight}$, to study how different weight ratios between belief alignment and task performance affect the uncertain agents' learning.} 

\para{Results} Figure~\ref{fig:manipulation} shows both quantitative and qualitative results. {\finalrevise 
The left plot in Figure~\ref{fig:manipulation} shows that under the \passive game algorithm, the certain agent (left robot) does not exploit the uncertain agent’s learning dynamics and thus fails to accelerate learning. In contrast, even with a low weight on intent demonstration, \ours enables the certain robot to actively influence the uncertain robot’s estimate dynamics, yielding faster convergence. As the weight increases, the certain agent deliberately \textit{exaggerates} its motion to make its goal more obvious, helping the uncertain agent quickly infer its intent and respond by moving directly toward the complementary pot handle (right, Figure~\ref{fig:manipulation} (b)), supporting~\textbf{H1}. 

Importantly, \ours computes these feedback intent-demonstration actions in real time, generating each action within 0.1 seconds. This efficiency enables strategic decision-making on hardware and ensures that a certain agent can robustly guide the learning of an uncertain agent, even under actuation noise, where the dynamics are less predictable.
}

{\revise{}\noindent \textit{\textbf{H2.} The pre-computed intent demonstration feedback policy $\bar{\pi}_t^1$ can efficiently convey the certain player’s unforeseen changes in its intent without the need of recomputing $\bar{\pi}_t^1$.}

\para{Setup and Metrics} We focus on the \textbf{Assistive Lunar Lander} environment. We measure the belief state and the physical state trajectories of the lunar lander under passive game and Algorithm~\ref{alg:LQ games}. For our method, we set $\rho_1 =1$ and $\rho_2=4$, to ensure the task performance and belief alignment. 

\para{Results} Although the feedback policy $\bar{\pi}_t^1$ is computed under the assumption that the certain agent’s intent remains stationary, it can still effectively convey unforeseen changes in certain agent's intent during deployment. Figure~\ref{fig:varying} demonstrates that the autopilot’s belief rapidly converges toward the initial target $\trueparam=25$ during the interval $t\in[0,20]$, and then adjusts efficiently to the updated human-preferred destination $\trueparam=50$ when $t\ge 20$. This highlights the robustness of the feedback policy $\bar{\pi}_t^1$ in realistic scenarios where task objectives shift unexpectedly during deployment—situations not explicitly considered when computing $\bar{\pi}_t^1$, yet handled effectively due to the feedback policy’s strong generalization to dynamically changing intents.
}



\noindent \textit{\textbf{H3.} The certain agent can improve its task performance by teaching agents with uncertainty.} 

{\revise{}\para{Setup and Metrics} We focus on the \textbf{Three-Vehicle Platooning} environment.  
We measure each player's task regret by comparing the optimal state-action trajectory $\xi^*$ under the \oracle game with the executed state-action trajectory $\xi$ under one of the incomplete information models: $\text{Regret}^i(\xi, \xi^*) := \sum^T_{t=0} [c^i(x_t, u_t) - c^i(x^*_t, u^*_t)]$. Lower regret indicates better performance. We set $\taskweight = 1$ and $\intentweight = 0$ to evaluate whether the policy $\bar{\pi}_t^1$ can strategically reduce the certain agent’s task cost when prioritizing task performance.}

\para{Results} 
{\finalrevise Figure~\ref{fig:regret} shows each player’s regret (y-axis) across all possible true intents of the certain player $\trueparam$ (x-axis). In this environment, \ours yields lower regret for the certain player (Player 1) compared to the \passive game baseline. This demonstrates that the certain agent can exploit the estimation dynamics of multiple uncertain agents to improve its task performance, confirming the scalability of our approach from two to more agents and providing direct support for \textbf{H3}.}


\noindent \textit{\textbf{H4.} When $\intentweight \equiv 0$, \ours balances task performance and intent demonstration for the certain agent. } 

{\revise{}\para{Setup and Metrics} We evaluate our method in the \textbf{Furniture Moving} environment, focusing on (1) the uncertain agent’s belief convergence and (2) the certain agent’s task cost. To test whether the certain agent can strategically influence belief without explicitly optimizing for intent demonstration, we set the intent demonstration hyperparameter to $\intentweight = 0$, thereby prioritizing task performance. We consider two true furniture angle preferences: $\trueparam = 0.3$ rad ($\sim 17^\circ$) and $\trueparam = 1.1$ rad ($\sim 63^\circ$). The furniture’s initial angle is always set to $\theta_0 = 0.6$ rad ($\sim 35^\circ$), and the uncertain agent’s initial belief is Gaussian, with mean $0.1$ and variance $0.4$.
}

\para{Results} 
{\finalrevise Even without explicit intent-demonstration terms in the cost, \ours enables the certain agent to influence the uncertain agent’s belief and achieve lower task cost than \passive game (Figure~\ref{fig:furniture} (a),(b)). While the real furniture angle always converges faster to $\trueparam$ under \ours (plot (c)), we observe that when $\trueparam=0.3$, the uncertain agent’s \textit{belief} converges more slowly than in the baseline (top plot (a)). This stems from the initial condition: since the initial angle $\param_0=0.6$ is already close to $\trueparam=0.3$, the certain agent conserves effort by prioritizing task completion over correcting the uncertain agent’s belief. When the initial and desired angles differ greatly, however, it becomes worthwhile for the certain agent to guide belief alignment to improve task performance. These strategic behaviors emerge automatically from the dynamic programming solution to the active intent-teaching problem, demonstrating our method’s ability to reason about the trade-off between task performance and intent demonstration, supporting \textbf{H4}.}








\section{Discussion}\vspace{-0.3em}
\para{Conclusion} In this work, we studied intent demonstration in multi-agent general-sum games, a problem commonly encountered in game-theoretic control applications such as autonomous driving, multi-robot manipulation, shared control systems, and human-robot interactions. 
Theoretically, we proved a sufficient condition for the convergence of an uncertain agent's beliefs to the ground truth certain agent's intent. Additionally, we showed that the certain agent could achieve a higher task performance by strategically demonstrating its intent to the uncertain agents. 
Algorithmically, we proposed an efficient method to solve linear and nonlinear intent demonstration problems via iterative linear-quadratic approximations.
Our empirical results show that intent demonstration accelerates the learning of uncertain agents, reduces task regret for players, and enables the certain agent to balance task performance with intent expression. \vspace{-0.3em}

{\color{black}\para{Limitations and Future Work} Our framework assumes a shared initial estimate and known estimate dynamics of the uncertain agents. While reasonable in some contexts (e.g., strong priors at a four-way intersection), model mismatch can be mitigated by computing $\bar{\pi}^1_t$ in a receding-horizon manner, where agents replan from the observed state at each step. In future work, the certain agent could further infer the uncertain agents' beliefs and estimate dynamics online and strategically influence the uncertain agents' beliefs.}\vspace{-0.3em}

\section*{Appendix}\vspace{-0.3em}
\label{app}
\begin{proofp}
\revise{
        We approach the proof by showing that there exists a teaching policy under which the belief $\est_t^j$, where $j\in\{2,\dots,N\}$, converges to the ground truth parameter exponentially fast. Substituting $u_t^1 = 
        \pi_t^1(x_t;\param)$ into player $j$'s estimate dynamics, we have $\est_{t+1}^j  = \est_t^j - \stepsize \nabla_{\est_t^j} \frac{1}{2}\|u_t^1 - \pi_t^1(x_t;\est_t^j)\|_2^2= \est_t^j + \stepsize K_{t,\param}^{1\top}K_{t,\param}^1 (\param - \est_t^j)$. 
    Subtracting $\param$, we have $\est_{t+1}^j - \theta= (I-\stepsize K_{t,\param}^{1\top}K_{t,\param}^1)(\est_t^j - \param) $. 
    Since 
    the largest singular value of $(I-\alpha K_{t,\param}^{1\top}K_{t,\param}^1)$, $\forall t\le T$, denoted as $c$, is strictly less than 1, we have $\|\est_{t+1}^j - \param\|_2 \le c \|\est_t^ j - \param \|_2$. 
    %
    Thus, there exists an active teaching policy, defined as $\bar{\pi}_t^1(x_t,\est_t;\theta) := \pi_t^1(x_t;\param)$, 
    that guarantees the exponential convergence of $\est_t^j$ towards $\trueparam$.}
\end{proofp}

\begin{proofpp}
    \ First, we observe that $\{\tilde{u}_t^1\}_{t=0}^{T}$ and $\{\tilde{x}_t\}_{t=0}^{T+1}$ constitute a feasible solution to \eqref{eq:active teaching problem}. The optimal solution of \eqref{eq:active teaching problem} leads to a cost value not greater than player 1's cost in complete information game. Moreover, when $\nabla_{u_{t:T}^1}\tilde{c}_{t:T}^1$ is nonzero at $\tilde{u}_{t:T}^1$, by convexity of the cost $c_t^1$ \cite[Section 4.2.3]{boyd2004convex}, for some $\epsilon>0$, there exists a solution $\check{u}_{t:T}^1 \in \{u_{t:T}^1: \|u_{t:T}^1 - \tilde{u}_{t:T}^1 \|_2\le \epsilon\}$ such that player 1's control $\check{u}_{t:T}^1$ achieves a lower task cost value $\tilde{c}_{t:T}^1$ than under the control $\tilde{u}_{t:T}^1$. 
\end{proofpp}

\bibliographystyle{ieeetr}
\bibliography{references}

\end{document}